\documentclass[acmsmall,nonacm]{acmart}
\AtBeginDocument{%
  \providecommand\BibTeX{{%
    \normalfont B\kern-0.5em{\scshape i\kern-0.25em b}\kern-0.8em\TeX}}}

\usepackage{amsmath}
\usepackage{dsfont}
\usepackage{algorithm}
\usepackage{algpseudocode}

\newcommand{\tx}{\texttt{txn}}
\newcommand{\user}{\texttt{USER}}

\newtheorem{theorem}{Theorem}
\newtheorem{prop}{Property}
\newtheorem{lem}{Lemma}

\begin{document}

\title{Masquerade: Simple and Lightweight Transaction Reordering Mitigation in Blockchains}

 \author{Arti Vedula}
 \email{vedula.9@osu.edu}
 \orcid{0000-0002-2465-352X}
 \affiliation{%
   \institution{The Ohio State University}
   \city{Columbus}
   \state{Ohio}
   \country{USA}
   \postcode{43210}
 }

 \author{Shaileshh Bojja Venkatakrishnan}
 \email{ bojjavenkatakrishnan.2@osu.edu}
 \affiliation{%
   \institution{The Ohio State University}
   \city{Columbus}
   \state{Ohio}
   \country{USA}
   \postcode{43210}
 }
 \author{Abhishek Gupta}
 \email{gupta.706@osu.edu}
 \affiliation{%
   \institution{The Ohio State University}
   \city{Columbus}
   \state{Ohio}
   \country{USA}
   \postcode{43210}
 }

 \renewcommand{\shortauthors}{Vedula et al.}

\begin{abstract}
  Blockchains offer strong security gurarantees, but cannot protect users against the ordering of transactions. Players such as miners, bots and validators can reorder various transactions and reap significant profits, called the Maximal Extractable Value (MEV). In this paper, we propose an MEV aware protocol design called Masquerade, and show that it will increase user satisfaction and confidence in the system. We propose a strict per-transaction level of ordering to ensure that a transaction is committed either way even if it is revealed. In this protocol, we introduce the notion of a "token" to mitigate the actions taken by an adversary in an attack scenario. Such tokens can be purchased voluntarily by users, who can then choose to include the token numbers in their transactions. If the users include the token in their transactions, then our protocol requires the block-builder to order the transactions strictly according to token numbers. We show through extensive simulations that this reduces the probability that the adversaries can benefit from MEV transactions as compared to existing current practices.
\end{abstract}





\maketitle

\section{Introduction}

Blockchains offer a decentralized and transparent platform for recording transactions with strong security guarantees.
Originally conceived as a payment system, in recent years smart contracts in blockchains have enabled the development of a plethora of decentralized applications (dapps) in myriad service sectors. 
Among these, dapps developed for financial use cases---such as lending and borrowing, stablecoins, exchanges, insurance etc.---have emerged as a popular alternative to centralized financial institutions. 
Colloquially referred to as de-fi (for decentralized finance), financial smart contracts are projected to grow to a half a trillion market~\cite{defisurge}.

Concurrent with the growth of defi dapps, transaction reordering attacks have become commonplace in blockchains. 
This is when a malicious actor attempts to confirm its own transaction ahead of a victim transaction on the blockchain {\em after} observing the victim transaction. 
A frontrunning attack is problematic to the victim as it can unfairly diminish the value earned by the victim in the transaction. 
For example, a victim looking to purchase tokens available at a lower-than-market price at a decentralized exchange may not get that price if the attacker purchases those tokens first. 
Unfortunately, blockchains by design facilitate transaction reordering attacks: unconfirmed transactions submitted by users are publicly visible in a `mempool’ and transactions can be prioritized for confirmation by increasing their transaction fees.

Outside of financially hurting a victim, reordering attacks can seriously affect the security of the blockchain consensus protocol especially in proof-of-stake systems. 
It is reported that in 2022 alone more than 300 million dollars were unfairly gained through reordering attacks~\cite{mevoutlook}. 
With such fortune at stake, it is feasible for an attacker to bribe block proposers to act in an unethical way or even deviate from protocol~\cite{chitra2022improving}.

Preventing transaction reordering attacks has therefore become an active area of research with several mitigation techniques proposed so far~\cite{yang2022sok}.
In practice, Ethereum has favored an approach where transaction order is determined by (trusted) third-party entities called builders. 
A builder collects transactions from end-users, orders them within a block and sends the block to a proposer for inclusion within the blockchain. 
Separating out the functionalities of ordering transactions in a block (building) and publishing a block (proposing) protects block proposers from being influenced to rearrange transactions.  
However, reordering attacks still very much happen at the builder level. 
We note that not all builders practice or tolerate reordering attacks, but a vast majority of them do.

In our work we present Masquerade, a simple and lightweight solution to preventing transaction reordering in blockchains. 
The key idea in Masquerade is the use of a token-number system (as at a takeout restaurant or the drivers’ license agency) for deciding transaction ordering. 
In Masquerade, a user first purchases a token with an assigned token number for a refundable fee. 
Subsequently the user can use this token to make a transaction. 
Transactions are ordered in the block in an increasing order of the token numbers of tokens associated with transactions. 
Once a token is used with a transaction, the funds paid for purchasing the token are returned to the user. 

The key intuition for why Masquerade works is that at the time of a token purchase, an attacker has no way of knowing what the purchased token will be used for in the future. 
An attacker is therefore forced to purchase as many tokens as possible, which places a significant financial stress on the attacker. 
Under a formal model of the system, if the total funds available to an attacker is a fraction $\sigma < 1$ of the total funds available to honest users, we show that the fraction of transactions that can be front run is at most $\sigma$. 
Moreover, Masquerade guarantees that attacker to user wealth ratio diminishes over time making the fraction of transactions front run asymptotically go to zero. 
Experiments on synthetic and real-world transaction data support our theoretical observations. 
We also note that Masquerade is robust to transaction reordering attacks occurring on token-purchase transactions submitted by users.

Deploying Masquerade in practice requires minimal changes to the consensus protocol---blocks must be verified to ensure transactions are in increasing order of their token numbers.
Issuance and maintenance of tokens can be easily implemented as a smart contract.  
In contrast, existing proposals enforcing some notion of a fair-ordering of transactions~\cite{kelkar2020order,kelkar2021themis,zhang2020byzantine,cachin2022quick}, or a commit-then-reveal scheme~\cite{khalil2019tex,bentov2019tesseract,zhang2022flash,malkhi2022maximal} require either significant changes to the consensus algorithm, rely on complex cryptographic primitives or a trusted party~\cite{heimbach2022sok}. 
Masquerade can also be implemented as an opt-in system in which only users seeking protection against reordering attacks participate. 
In summary, the contributions of this paper are: 
\begin{itemize}
\item We propose a simple, lightweight solution that uses a token-based ordering mechanism for mitigating transaction reordering attacks in blockchains such as Ethereum. 
\item We rigorously show performance of Masquerade under a theoretical model. 
\item We evaluate Masquerade experimentally using a variety of synthetic and real-world data to corraborate our theoretical results. 
\end{itemize}

\section{Background}

\subsection{Maximal Extractable Value (MEV)}

With a push towards decentralization, there has been an increase in decentralized finance (DeFi) protocols that are run by smart contracts with individual users simply interacting with them. These are permissionless, transparent and eliminate the need for a central authority. DeFi services such as Uniswap\cite{adams2021uniswap}, Compound\cite{yang2022sok}, MakerDAO\cite{brennecke2022central} etc are based on open protocols and Decentralized Applications(DApps) and aim to recreate traditional financial systems and services without the need for a centralized authority. 
In DeFi dapps, the order in which transactions are executed play a very important role in profits gained by users of this system. 
Profits obtained by bad actors through transaction reordering attacks is called as maximal extractable value (MEV). 
MEV profits are in the millions of dollars today. 
As a result, front running victim transactions is a very lucrative option. 

In centralized finance markets (CeFi) such as Binance\cite{noauthor_binance_nodate}, CoinBase\cite{kazerani2017determining}, Kraken \cite{pichl2017volatility} etc, transaction ordering is strictly monitored by an authority, and adversaries are punished. 
A major problem with these CeFi protocols is that the users are required to relinquish control of their assets to the service provider and trust the operator while transacting.  
A canonical example of a DeFi app is a decentralized exchange (DEX), which is a platforms that allows trading through smart contracts on the blockchain. 
DEX transactions are particularly susceptible to front running attacks by adversaries, as there is a lack of any regulatory body that can punish these malpractices. 
Further, clients who have been manipulated by this exchange have no way of recovering their lost wealth.      

An Automated Market Maker(AMM) is a smart contract used in DEXs to facilitate the trading of digital assets without the need for traditional order books or other intermediaries. These AMMs use some predetermined formulas to determine the prices of assets based on their supply and demand within a liquidity pool. AMMs were introduced as an alternative to traditional order book exchanges, where buyers and sellers interact directly to set the prices of assets. An AMM contains a liquidity pool of atleast two assets, where it allows users to deposit an asset and withdraw the other based on the determined exchange value.
One of the most commonly used AMMs is the Constant Product Market Maker(CPMM) algorithm. This algorithm is used in DEXs like Uniswap\cite{adams2021uniswap} and PancakeSwap\cite{noauthor_pancakeswap_nodate}. In the CPMM algorithm, the product of the quantities of two assets in a pool remains constant, which helps determine the price when one asset is traded for another. Simply, if an exchange has $x$ and $y$ units of currency in their reserve, $xy=k$ denotes the price of the asset using the exchange.

\subsection{MEV Attacks}
Currently, in blockchains such as Ethereum, adversaries frequently monitor the mempool to search for a certain kind of MEV transactions which would yield reward at the expense of increased cost to the user. We explain three major attacks, benefits to the adversary, and the cost to the user below. 

\begin{enumerate}
    \item Sandwich Attack: In this case, an attacker attempts to monitor mempool for pending transactions trading large assets that may cause price fluctuations. Let a user transaction be $\tx_{u}$ that is exchanging currency $C_1$ to $C_2$, denoted by $C_1 \xrightarrow[]{}C_2$. An attacker $a$ then submits two transactions, say $\tx_{a}$ and $\tx_a'$, with the same transaction amount as the user $u$. In $\tx_{a}$, the attacker exchanges $C_1\xrightarrow{}C_2$, and in $\tx_a'$ exchanges $C_2 \xrightarrow[]{}C_1$. A sandwich attack happens if $\tx_{a}$, $\tx_{u}$ and $\tx_{a}'$ are included in the same block $\mathcal B$ and in this order.

    \item Arbitrage MEV: An arbitrage occurs when users are trading assets across different exchanges. Let a user transaction be $\tx_{u}$ that is trading $\alpha_u$ amount of currency $C_1$ with fiat money across two different exchanges $E_1$ and $E_2$ with different exchange rates $P_{C_1,E_1}$ and $P_{C_1,E_2}$, respectively. A successful arbitrage occurs if price $\alpha_u P_{C_1,E_1}+g_{u}<\alpha_u P_{C_1,E_2}$, for gas fees $g_u$. In this case, the attacker $a$ monitors the public mempool for an arbitrage opportunity and attempts to front-run $\tx_{u}$ with $\tx_{a}$ with $\alpha_a = \alpha_u$, gas fee $g_{a}>g_{u}$, so that $\alpha_{a}P_{C_1,E_1}+g_{a}<\alpha_{a}P_{C_1,E_2}$. In this case, the arbitrage benefit to the user is diminished and the adversary gains from front-running such a transaction.

    \item Liquidation MEV: A liquidation attack occurs on a loan taken out by a user. Let a user borrow $\alpha_{C_1}$ units of currency $C_1$, for a price of $\alpha_{C_1}C_1$. The user in exchange offers a collateral of units $\alpha_{C_2}$ of currency $C_2$ for a price of $\alpha_{C_2}C_2\geq \alpha_{C_1}C_1$. This attack occurs, when the collateral $\alpha_{C_2} C_2$ no longer covers the value of the debt. 
    Let the exchange rate at the time of borrowing be $C_2 = \beta_r C_1$. For the loan to be healthy, $\beta>1$. Due to price fluctuations, it is possible that $\frac{\alpha_{C_2}C_2}{\alpha_{C_1}C_1}<1$ which means that the loan is now under-collateralized. At this time, the collateral is available to users to purchase at a low rate. An adversary now frontruns the purchase of the collateral $\alpha_{C_2}C_2$, and is able to accrue a profit on this transaction by selling it later at a higher price.  
\end{enumerate}

\subsection{MEV Mitigation: Status Quo}
Ethereum uses a proposer-builder separation architecture for mitigating the negative impacts of MEV on blockchain security~\cite{heimbach2023ethereum}. 
Under this architecture, the process for the formation of a block is as follows. 
Users submit transactions either by broadcasting publicly over the blockchain network, or by sending privately to a third-party reputable entity called a builder (Flashbots, BeaverBuild, Builder0x69, BloXroute etc.).
Transactions made include an appropriate amount of transaction fees depending on the priority desired for the transaction. 
The builder collects transactions and orders them to create a candidate block which it then advertises to the proposer of that time slot.
Builder strive to construct blocks with high aggregate transaction fees, as a portion of the transactions fees goes to the builder.
Competing builders advertise blocks to the proposer. 
The proposes chooses the block containing the most amount of fees for publication on the blockchain.


\section{Model}
We consider time is divided in to discrete rounds, with one block produced during each round. 
Our model consists of two actors, a user and an adversary as described in the following. 

\smallskip 
\noindent
    {\bf User:} The user in our model represents collectively all honest users in the system. 
    The essence of our results does not change even if we explicitly consider multiple users in the model. 
    We assume the user is honest and follows protocol. 
    The user seeks to make MEV transactions without getting front run to maximize its profits. 
    We define an ``MEV transaction" as a DeFi transaction from which value can be extracted by the adversary through a reordering attack. 
    The user is interested in making at most one MEV transaction each round. 
    The profit gained by the user upon making a MEV transaction successfully (i.e., without getting front run) and the profits lost when a front running attack happens are discussed in \S\ref{s:modelrewards} below. 
    At the beginning of the experiment, the user has a net wealth of $W_u[0]$. 
    

    \smallskip 
    \noindent
    {\bf Adversary:} The adversary in our model is an entity that seeks to make profit by attacking the user's MEV transactions. 
    We primarily consider front running attacks in our work, though the results extend to back running and sandwich attacks as well. 
    During an MEV attack, the adversary gains precisely the value lost by the user on the transaction. 
    At the beginning of the experiment, we assume the adversary has a net wealth of $W_a[0]$. 
    We assume the total wealth $W_a[0]$ of the adversary is lower than the total wealth of the user $W[0]$ initially by a factor of $\sigma < 1$. 
    This is a reasonable assumption, as the security of Proof-of-Stake consensus followed by Ethereum relies on such an assumption as well. 
    For any round $r$, we let $W_u[r], W_a[r]$ respectively denote the total wealth of the user and adversary respectively at the beginning of round $r$. 
    Note that since the user can make at most one MEV transaction each round, the adversary can also front run at most one MEV transaction each round. 
    We do not consider regular (i.e., non-MEV) transactions made by the user or the adversary in our model. 
    The adversary has complete knowledge of the internal state of the user at any time. 
    Unless the proposed consensus protocol prohibits it, the adversary can front run any transaction submitted by the user during a round with its own transaction.

\smallskip
The blockchain network also contains builders, relays and validators, but for our problem we do not consider them to be an essential part of the network dynamics, and omit their roles. 
We define a transaction as $\tx_{u}[r]$ for a transaction made by a user at round $r$, and $\tx_{a}[r]$ for a transaction made by an adversary.



\subsection{Rewards}
\label{s:modelrewards}
We assume that the profit made by an honest user on an MEV transaction is $\eta$, of which he loses $f\eta$ if an adversary manages to front run the transaction. 
In practice, users can specify a slippage parameter to control their MEV loss which relates to the $f$ in our model~\cite{heimbach2022eliminating}. 
Thus, the rewards to the user and the adversary, respectively, $h_{u}[r],h_{a}[r]$ in round $r$ can be defined as follows:
\begin{align}
h_{u}[r] &= \begin{cases}
\eta-f\eta&\text{if MEV transaction is front run}\\
\eta&\text{otherwise},
\end{cases}  \\
h_{a}[r] &= \begin{cases}
f\eta&\text{if MEV transaction is front run}\\
0&\text{otherwise}.
\end{cases}  
\end{align}

We ignore the costs incurred by gas fees to the user and the adversary. 
In today's Ethereum (referred to as ``current protocol" in the paper), we assume a user's MEV transaction always get front run which leads to a profit of $\eta-f\eta$ per round for the user. 
This is a reasonable assumption, as a user today either issues its transaction publicly and gets attacked, or issues its transaction privately to a builder by paying hefty fees. 
Either way the value the user rightfully must gain in the transaction is lost in today's Ethereum. 
After $R$ rounds, the total wealth accumulated by  the user and the adversary in the current protocol from MEV transactions is given by:
\begin{align}
   W_u &=W_{u}[0]+\sum_{r=1}^{R}h_{u}[r] \\
     &=W_{u}[0]+(\eta-f\eta)R \\
   W_a & =W_{a}[0]+\sum_{r=1}^{R}h_{a}[r] \\
     &=W_{a}[0]+f\eta R
\end{align}

Thus, we see that an honest user is losing out on atleast $f\eta R$ profits on having made an MEV transaction every round for $R$ rounds. In the real world scenario, they end up losing even more money when  multiple MEV transactions are part of a block. 

\subsection{Problem Statement}


Our objective is to design a transaction ordering protocol that prevents  MEV attacks and maximizes the total wealth of the after $R$ rounds. 
Or, equivalently, we would like to reduce the fraction of transactions that are successfully attacked by the adversary. 
The solution space we explore must obey the following constraints. 
First, we do not want to introduce any significant modifications to the consensus protocol keeping in mind the difficulties involved in implementation. 
Any solution we propose must be implementable with just a few lines of code, either at the consensus or execution layers. 
We also avoid use of computationally expensive cryptographic algorithms due to their complexity of implementation. 
Finally, we would like our method to be resistant to attacks without the usage of any trusted third parties. 
                                                                                            
\section{Proposed Transaction Ordering Protocol}
 We introduce our transaction ordering protocol, called Masquerade, which is a decentralized protocol with minimal changes to the current consensus protocol and no reliance on external trusted parties. In our protocol, we modify a percentage of transactions $\tx_u,\tx_a$ to include a new parameter called "token number" with them. These transactions will then be ordered strictly according to the token number accompanying the transaction. As a result, even if the content of the MEV transaction is made aware to an attacker, they are unable to frontrun it without the relevant token number required to actually frontrun the transaction.



We add two new kinds of transactions called "token purchase transaction" and "tokenized transaction" that we will describe in more detail below: 

\subsection{Token Purchase Transaction}
A token purchase transaction is simply a transaction that the user makes in order to receive token number $T$ that can be used for future transactions. A single token purchase transaction can be used to specify any number of tokens desired, as long as sufficient funds are available for the token. A single token purchase costs $y$ units. A token purchase transaction is considered successful, if it has been included as a part of the main chain. A user, or adversary cannot specify the token number they would like to purchase. Token numbers are issued independently, by a token issuing algorithm. If a user uses a valid token, to make a valid tokenized transaction, the token is considered to be spent, and the cost of purchasing that particular token $y$ is refunded back to the user. Each token can only be used once, but a token that has been purchased once never expires, and can be used in the future. 

\subsection{Tokenized Transaction}
A tokenized transaction for round $r$, $\tx_{u}[r],\tx_{a}[r] $ is a transaction that is accompanied by a valid token number $T_{u}[r],T_{a}[r]$ for the user and adversary respectively. A valid token number is one, that has been confirmed and included on the main chain in the previous rounds $r-1$. A user can only use valid tokens to make a tokenized transaction. Tokenized transactions are strictly ordered in ascending order, with transaction $\tx_{u}[r]$ being executed earlier than $\tx_{a}[r]$, if $T_{u}[r]<T_{a}[r]$. Further, all tokenized transactions precede non-tokenized transactions.  


The process of formation of the block now follows the following procedure:
\begin{itemize}
    \item  At the beginning of the round, a user can take the following actions:
    \begin{itemize}
        \item all users can make non-tokenized regular transactions or non-tokenized MEV transactions.
        \item  users can make the desired token purchase requests.
        \item users can use previously purchased tokens to make tokenized MEV transactions, or tokenized regular transansactions. 
    \end{itemize}
    \item All transactions are verified to ensure they are valid.
    \item The validator for round $r$ then collects $N$ valid tokenized transactions, and orders them based on token numbers. A tokenized transaction $\tx_{u}[r]$ is strictly ordered before $\tx_{a}[r]$ if $T_{u}[r]<T_{a}[r]$.
    \item The validator also collects non-tokenized transactions and creates a block based on the highest rewards that can be extracted.
    \item Finally, at the end of the round, the validator updates the state of the Blockchain and publishes the block.
\end{itemize}

It is important to note, here, that the token purchase transaction, like any other transaction, has the ability to get attacked. An adversary is able to frontrun an honest user's token purchase transaction, to get a lower token number than the user. 

\subsection{Rewards for the proposed protocol}
During round $r$, there are only $l$ tokens available that can be purchased. We assume a reasonable, fixed user policy $\pi_u$ and adversary policy $\pi_a$. We would like to find the maximum possible reward that can be achieved by the adversary in a fixed user policy scenario, which is strictly less than what the adversary gains in the current scenario. Let the user and adversary each purchase $\mathsf X_{u}[r], \mathsf X_{a}[r]$ number of tokens respectively. As each token costs $y$ units, the total costs incurred by them is $y\mathsf X_{u}[r], y\mathsf X_{a}[r]$ Let each block only allow $N$ tokenized transactions.

Let us define $\mathsf M_{u}[r]$ to be a function that defines whether an MEV transaction is made by a user or not i.e

\begin{equation}
\mathsf M_{u}[r] = \begin{cases}
1&\text{if MEV transaction is made}\\
0&\text{otherwise}
\end{cases}  
\end{equation}

Similarly, we define $\mathsf M_{a}[r]$ to be a function that defines whether an MEV transaction is attacked by an adversary or not i.e

\begin{equation}
\mathsf M_{a}[r] = \begin{cases}
1&\text{if MEV transaction is attacked}\\
0&\text{otherwise}
\end{cases}  
\end{equation}

Let us define $\mathsf F_r$ to be a function that defines whether an MEV transaction made by the user is frontrun or not i.e

\begin{equation}
\mathsf F_{r} = \begin{cases}
1&\text{if MEV transaction is frontrun}\\
0&\text{otherwise}
\end{cases}  
\end{equation}

The rewards to the user and the adversary $h_{u}[r],h_{a}[r]$ in round $r$ can be defined as follows:
\begin{equation}
h_{u}[r] = \begin{cases}
-y\mathsf X_{u}[r]+\eta-f\eta+y&\text{if MEV transaction is frontrun}\\
-y\mathsf X_{u}[r]+\eta+y&\text{otherwise}
\end{cases}  
\end{equation}
\begin{equation}
h_{a}[r] = \begin{cases}
-y\mathsf X_{a}[r]+f\eta+y&\text{if MEV transaction is frontrun}\\
-y\mathsf X_{a}[r]&\text{otherwise}
\end{cases}  
\end{equation}

Now, both the user and adversary start with an initial wealth $W_{u}[0], W_{a}[0]$. If, we assume that each block has only a single MEV transaction, after $R$ rounds, the maximum wealth that can be accumulated by both the user and the adversary from MEV transactions are:

\begin{align}
   W_u&=W_{u}[0]+\max\sum_{r=1}^{R}h_{u}[r]\notag\\
     &=W_{u}[0]+\max\sum_{r=1}^{R}\eta-f\eta\mathsf F_r+y\mathsf M_{u}[r]-y\mathsf X_{u}[r]
\end{align}
  
\begin{align}
   W_a&=W_{a}[0]+\max\sum_{r=1}^{R}h_{a}[r]\notag\\
     &=W_{a}[0]+\max\sum_{r=1}^{R}(f\eta+y)\mathsf F_r-y\mathsf X_{a}[r]
\end{align}

We now introduce a randomly chosen fixed user policy $\pi_u$, as shown in Algorithm 1. At the end of each round, a user purchases a single token, which depletes their wealth. When the threshold of this wealth is less than a small threshold $\tau$, the user begins to spend their tokens. This ensures that there is sufficient time for user to collect tokens to spend so that they are not frontrun by adversary tokens. Finally, the user will always take the MEV opportunity presented to them as long as they have an appropriate token. We assume that an adversary can simply frontrun any non-tokenized transactions. As a result, the user solely makes tokenized MEV transactions. Further, as the user is not aware of any tokens that are held by the adversary, they always use their lowest available token.

\begin{algorithm}
\caption{User Policy $\pi_u$}\label{Algorithm 1}
\begin{algorithmic}[1]
\State Inputs: $W_{u}[r],y,\tau$, $\tilde H_{u}[r]\subseteq \tilde{H}[r]$
\If{$W_{u}[r]>y$}
    \State $\mathsf X_{u}[r]=1$
\EndIf
\If{$W_{u}[r]\leq \tau$}
    \State $M_{u}[r]=1$
    \State $T_{u}[r]=\tilde H_{u}[r]$
\Else
    \State $M_{u}[r]=0$
    \State $T_{u}[r]=\infty$
\EndIf
\end{algorithmic}
\end{algorithm}

The adversary policy  $\pi_a$ on the other hand, is more powerful. As a worst case scenario, we assume that the adversary has full control over the token number assignment, and may reorder the token purchase transactions as they please. Thus, the adversary is aware of all tokens that the user has. We also allow the adversary to monitor user transactions, and the ability to frontrun these transactions. Since the user only makes tokenized MEV transactions, the adversary is able to frontrun these transactions, if they have tokens that are smaller in number than user tokens. This policy is detailed in Algorithm 2. The adversary can choose any policy, however we show in Section 5, that this indeed is the best policy that can be taken. 

\begin{algorithm}
\caption{Adversary Policy $\pi_a$}\label{Algorithm 2}
\begin{algorithmic}[1]
\State Inputs: $y, W_{a}[r], T_{u}[r], M_{u}[r], \tilde H_{a}[r]\subseteq \tilde{H}[r]$
\State $\mathsf X_{a}[r]=\lfloor\frac{W_{a}[r]}{y}\rfloor$
\If{$M_{u}[r]=1$}
    \If{$T_{u}[r]<\infty$}
    \State $T_{a}[r]=\max T\in \tilde H_{a}[r]
    $ s.t. $T<T_{u}[r]$
    \Else
    \State $T_{a}[r]=\infty$
\EndIf
    \State $M_{a}[r]=0$
    \State $T_{a}[r]=\infty$
\EndIf
\end{algorithmic}
\end{algorithm}


\section{Analysis}

To ease analysis, we divide time in to epochs as defined in the following. 
The first epoch begins at round $r=0$ and ends when the user has completed the initial token purchasing as in Algorithm~\ref{Algorithm 1}. 
Equivalently, the first epoch lasts until the user's wealth $W_u[r]$ drops below threshold $\tau$. 
Each subsequent epoch begins immediately after the epoch prior to it ends. 
An epoch ends when the following two conditions are satisfied: 
\begin{enumerate}
    \item the user has utilized all of the tokens purchased in the previous epoch; 
    \item the available wealth $W_u[r]$ of the user drops below $\tau$. 
\end{enumerate}
The above conditions lead to a well-defined notion of an epoch as the user always utilizes the earliest available token for a MEV transaction.
Tokens in one epoch, therefore, are completely utilized before tokens in the next epoch are utilized. 
We also define a terminal epoch in which parties utilize tokens from the previous epoch, but make no new token purchases during the epoch. 
The game ends after the terminal epoch. 
We assume the game lasts for $k > 0$ epochs. 
Restricting the game to $k$ epochs captures the intuition that in practice, a player is  interested in optimizing her  rewards over a fixed time horizon, e.g., the lifespan of the player, the next five years etc.
Note that our definition of an epoch is tied to the user's behavior described in Algorithm~\ref{Algorithm 1}. 
The definition is independent of the adversary's behavior. 

Let $\tilde{H}[e]$ be the set of tokens purchased by either the user or the adversary during epoch $e$. 
Further, let $\tilde{H}_u[e] \subseteq \tilde{H}[e]$ and $\tilde{H}_a[e] \subseteq \tilde{H}[e]$ be the tokens purchased by the user and adversary, respectively, during $e$. 
$\tilde{W}_a[e]$ and $\tilde{W}_u[e]$ denote, respectively, the total wealth of the adversary and user at the end of epoch $e$. 
The total wealth of a party includes the wealth available for spending and the funds locked up in the form of tokens. 
We call a policy $\pi_a$ followed by the adversary as balanced if the following holds true.

\begin{figure}
    \centering
    \includegraphics[width=0.9\linewidth]{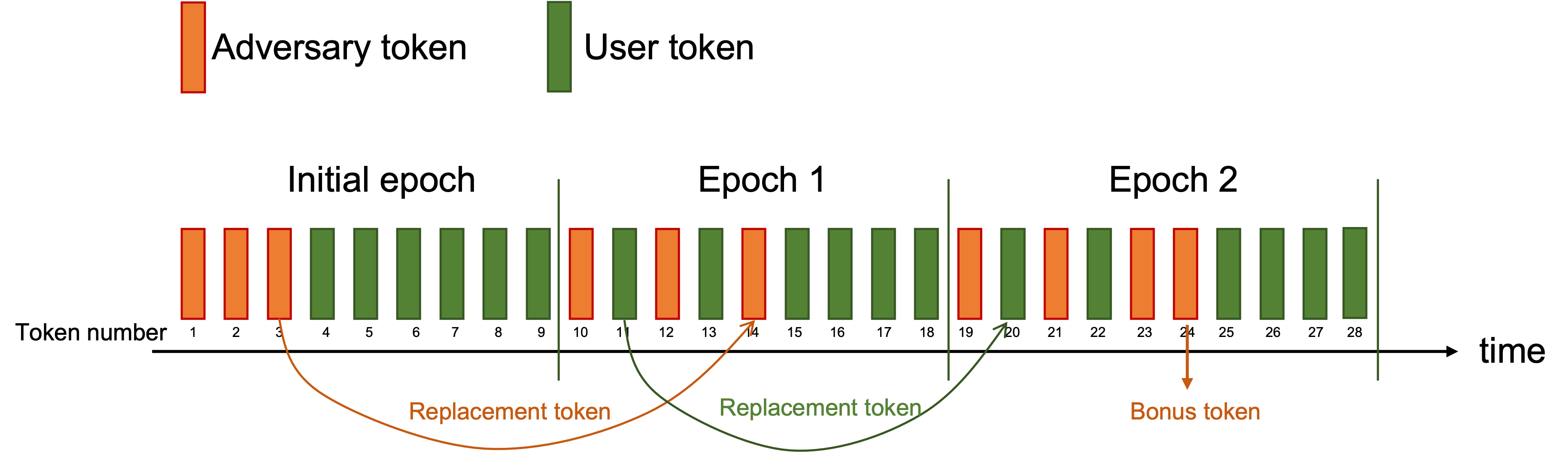}
     \caption{Illustration of token purchases with time.}
    \label{fig:time diagram}
\end{figure}

\begin{prop}[Balanced policy]
An adversary policy $\pi_a$ is balanced if for each non-terminal epoch $e$, there exists an injective mapping $m_e: \tilde{H}_a[e] \rightarrow \tilde{H}_u[e]$ such that $m_e(T) > T$ for all $T \in \tilde{H}_a[e]$ and $|\tilde{H}_a[e]| = \lfloor \tilde{W}_a[e]/y \rfloor$. 
\end{prop}
With a balanced policy, an adversary can successfully utilize each purchased token  to front run some victim  transaction.  
The adversary also maximally utilizes its wealth to purchase tokens. 
We call a specific epoch $e$ as balanced if there exists an injective mapping $m_e: \tilde{H}_a[e] \rightarrow \tilde{H}_u[e]$ such that $m_e(T) > T$ for all $T \in \tilde{H}_a[e]$ and $|\tilde{H}_a[e]| = \lfloor \tilde{W}_a[e]/y \rfloor$. 
We now show that the adversary policy $\pi_a$ described in Algorithm~\ref{Algorithm 2} is balanced. 
\begin{theorem}
For $W_a[0] < \sigma W_u[0]$ where $0 < \sigma < \frac{1}{2}$ is a security parameter, $f < \frac{1-\sigma - \epsilon}{1+\sigma}$, and $W_a[0] > \frac{y^2}{\eta \epsilon}$ where $\epsilon < (yf + f^2 \eta)/(\eta(1-f))$ is a parameter and $\tau < \epsilon W_u[0]$ the adversary's policy $\pi_a$ as described in Algorithm~\ref{Algorithm 2} is balanced. 
\end{theorem}
\begin{proof}
The proof is by induction. 
During the initial epoch, the adversary purchases $\lfloor W_a[0]/y \rfloor$ tokens first following which the user purchases $\lceil (W_u[0]-\tau)/y \rceil$ tokens as per the adversary policy $\pi_a$ and user policy $\pi_u$ in Algorithm~\ref{Algorithm 1}.
By assumptions in Theorem, $W_a[0] < \sigma W_u[0] < \frac{W_u[0]}{2}$ and therefore $\lfloor W_a[0]/y \rfloor < \lceil (W_u[0]-\tau)/y \rceil$. 
The injective mapping $m_0: \tilde{H}_a[0] \rightarrow \tilde{H}_u[0]$ can simply be $m_0(i) = \lfloor W_a[0]/y \rfloor + i$ for all $1 \leq i \leq \lfloor W_a[0]/y \rfloor$.
The total wealth of the user $\tilde{W}_u[0]$ and adversary $\tilde{W}_a[0]$ at the end of the initial epoch are still $W_u[0]$ and $W_a[0]$ respectively. 

Now, consider an epoch $e$ that is not the initial or the terminal epoch. 
Supposing the previous epoch $e-1$ is balanced, i.e.,  satisfies the property that there exists an injective function $m_{e-1}: \tilde{H}_a[e-1] \rightarrow \tilde{H}_u[e-1]$ with $m_{e-1}(T) > T$ for all $T \in \tilde{H}_a[e-1]$.
Also suppose the mapping $m_{e-1}$ is such that there are at least $c \lceil \tilde{W}_a[e-1]/y \rceil$ user tokens at the end of epoch $e-1$ without a preimage in $m_{e-1}$ where $c>(f+\epsilon)/(1-f-\epsilon)$, i.e., 
\begin{align}
|\{ T' \in \tilde{H}_u[e-1]: T' > \max_{T: T\in \tilde{W}_a[e-1]} m_{e-1}(T) \}| > c\left\lceil \frac{\tilde{W}_a[e-1]}{y} \right\rceil.  \label{eq:lowerboundontailend}
\end{align}
Equation~\eqref{eq:lowerboundontailend} implies there exist at least $c \lceil \tilde{W}_a[e-1]/y \rceil$ user transactions that are not front run in the end of epoch $e$. 
Further, suppose that the total wealth of the parties are such that $\tilde{W}_a[e-1] < \sigma \tilde{W}_u[e-1]$. 
We will show that there exists an injective function $m_e: \tilde{H}_a[e] \rightarrow \tilde{H}_u[e]$ with $m_e(T) > T$ for all $T \in \tilde{H}_a[e]$ and $\tilde{W}_a[e] < \sigma \tilde{W}_u[e]$. 

There are $\lfloor \tilde{W}_a[e-1]/y \rfloor$ tokens purchased by the adversary during epoch $e-1$. 
Since epoch $e-1$ is balanced, during epoch $e$ the adversary can use each of its tokens in $e-1$ to successfully front run transactions. 
After each successful front running attack, the adversary replaces its used token in epoch $e-1$ by purchasing a new token in epoch $e$.
We call such a token the adversary purchases as a {\em replacement} token. 
Additionally, the adversary can use the profits gained from front running to purchase tokens whenever the adversary's balance exceeds $y$. 
We call such a token as a {\em bonus} token. 
Note that after each transaction the user submits, the user also purchases a replacement token and possibly bonus tokens. 
We define replacement token and bonus token for the user analogous to those for the adversary. 
Per our model, the replacement and bonus tokens of the adversary are purchased before the replacement and bonus tokens of the user after each front running attack. 

To construct an injective mapping for epoch $e$, first note that the replacement token purchased by the adversary after each successful front run can be mapped to the replacement token purchased by the user after the front run. 
Thus, $\lfloor \tilde{W}_a[e-1] / y \rfloor$ of the adversary's tokens can be mapped to higher tokens in epoch $e$. 
However, in addition to replacement tokens the adversary also purchases $\lfloor \lfloor \tilde{W}_a[e-1] / y \rfloor f\eta / y \rfloor$ bonus tokens in epoch $e$, which need to be mapped. 
For this, consider round $r^*$ when the last front run attack is performed by the adversary in epoch $e$. 
We observe that, by policy $\pi_a$ all of the adversary's bonus tokens are purchased by round $r^* + 1$ in epoch $e$ (Figure~\ref{fig:time diagram}).
Transactions submitted by the user after round $r^*$ are not attacked. 
By assumption, at least $\lceil \tilde{W}_a[e-1]/y \rceil$ transactions are submitted by the user after round $r^*+1$.
Each of those transactions generate a reward of $\eta$ for the user. 
Therefore, the user is able to purchase at least $\lfloor \lceil \tilde{W}_a[e-1]/y\rceil (y + \eta)/y \rfloor$ tokens at the end of epoch $e$. 
To complete our mapping, we map the $\lfloor \lfloor \tilde{W}_a[e-1] / y \rfloor f\eta / y \rfloor$ bonus tokens of the adversary to the first $\lfloor \lfloor \tilde{W}_a[e-1] / y \rfloor f\eta / y \rfloor$ user tokens purchased after round $r^*$. 
This is possible since 
\begin{align}
\epsilon &< \frac{yf + f^2\eta}{\eta(1-f)} 
 \Rightarrow \frac{\eta \epsilon  + f\eta }{y + \eta} < \frac{f}{1-f} 
\Rightarrow \frac{\frac{y^2}{\tilde{W}_a[e-1]} + f\eta }{y + \eta} 
 < \frac{f}{1-f} \notag \\
\Rightarrow \frac{\frac{y^2}{\tilde{W}_a[e-1]} + f\eta }{y + \eta} &< c  
\Rightarrow y^2 < \tilde{W}_a[e-1](c y + c \eta-f\eta)\notag 
\end{align}
\begin{align}
\Rightarrow \tilde{W}_a[e-1] f\eta &< c \tilde{W}_a[e-1](y + \eta) - y^2 
\Rightarrow  \frac{\tilde{W}_a[e-1] f\eta}{y^2}  <   \frac{c \tilde{W}_a[e-1] (y + \eta)}{y^2} - 1 \notag \\
\Rightarrow \lfloor \lfloor \tilde{W}_a[e-1] / y \rfloor f\eta / y \rfloor &< \lfloor c \lceil \tilde{W}_a[e-1]/y\rceil (y + \eta)/y \rfloor. 
\end{align}
Thus, we have showed the existence of an injective mapping $m_e: \tilde{H}_a[e] \rightarrow \tilde{H}_u[e]$ with $m_e(T) > T$ for all $T \in \tilde{H}_a[e]$. 

The number of user tokens at the end of epoch $e$ who which a preimage on $m_e$ does not exist is at least $\lfloor \lceil \tilde{W}_a[e-1]/y\rceil (y + \eta)/y \rfloor - \lfloor \lfloor \tilde{W}_a[e-1] / y \rfloor f\eta / y \rfloor$. 
In the following we show that this quantity is at least $\lceil \tilde{W}_a[e]/y \rceil$. 
We have
\begin{align}
\frac{f + \epsilon}{1- f - \epsilon} &< c \notag \\
\Rightarrow \frac{f + \frac{y^2}{\tilde{\eta W_a[e-1]}}}{1-f - \frac{ y^2}{\eta \tilde{W}_a[e-1]}} &< c \notag  
\end{align}
\begin{align}
\Rightarrow f\eta + \frac{y^2}{\tilde{W_a[e-1]}} &< c\eta(1-f) -  \frac{c y^2}{\tilde{W}_a[e-1]}\notag \\ 
\Rightarrow cy + c f\eta + \frac{c y^2}{\tilde{W}_a[e-1]} &< c y + c \eta - \frac{y^2}{\tilde{W_a[e-1]}} - f\eta \notag 
\end{align}
\begin{align}
\Rightarrow c\frac{\tilde{W}_a[e-1]}{y}\left(1 + \frac{f\eta}{y} \right) + c &< \frac{c \tilde{W}_a[e-1](y+\eta)}{y^2} - 1 - \frac{\tilde{W}_a[e-1]f\eta}{y^2} 
 \notag \\ 
\Rightarrow c\frac{\tilde{W}_a[e]}{y} + c &< \frac{c \tilde{W}_a[e-1](y+\eta)}{y^2} - 1 - \frac{\tilde{W}_a[e-1]f\eta}{y^2} \notag  \\
\Rightarrow c \tilde{W}_a[e]/y + c &< c \lceil  \tilde{W}_a[e-1]/y\rceil (y + \eta)/y - 1 -\lfloor \tilde{W}_a[e-1] / y \rfloor f\eta / y \notag \\ 
\Rightarrow c \lceil \tilde{W}_a[e]/y \rceil &<  \lfloor c \lceil \tilde{W}_a[e-1]/y\rceil (y + \eta)/y \rfloor - \lfloor \lfloor \tilde{W}_a[e-1] / y \rfloor f\eta / y \rfloor 
\end{align}

It only remains to show that $\tilde{W}_a[e] < \sigma \tilde{W}_u[e]$. 
We have, 
\begin{align}
\frac{\tilde{W}_a[e-1] + \left\lfloor \frac{\tilde{W}_a[e-1]}{y} \right\rfloor f \eta}{\tilde{W}_u[e-1] + \left\lfloor \frac{\tilde{W}_a[e-1]}{y} \right\rfloor (1-f) \eta + \left( \left\lceil \frac{\tilde{W}_u[e-1] - \tau}{y} \right\rceil - \left\lfloor \frac{\tilde{W}_a[e-1]}{y} \right\rfloor \right) \eta } < \sigma \notag \\
\Leftarrow \frac{\tilde{W}_a[e-1] + \left\lfloor \frac{\tilde{W}_a[e-1]}{y} \right\rfloor f \eta}{\tilde{W}_u[e-1] + \left\lfloor \frac{\tilde{W}_a[e-1]}{y} \right\rfloor (1-f) \eta + \left( \left\lceil \frac{\tilde{W}_u[e-1] - \tau}{y} \right\rceil - \left\lfloor \frac{\tilde{W}_a[e-1]}{y} \right\rfloor \right) \eta } < \frac{\tilde{W}_a[e-1]}{\tilde{W}_u[e-1]} \notag \\
\Leftarrow \frac{\tilde{W}_a[e-1] + \left\lfloor \frac{\tilde{W}_a[e-1]}{y} \right\rfloor f \eta}{\tilde{W}_a[e-1]} < \frac{\tilde{W}_u[e-1] + \left\lfloor \frac{\tilde{W}_a[e-1]}{y} \right\rfloor (1-f) \eta + \left( \left\lceil \frac{\tilde{W}_u[e-1] - \tau}{y} \right\rceil - \left\lfloor \frac{\tilde{W}_a[e-1]}{y} \right\rfloor \right) \eta}{\tilde{W}_u[e-1]} \notag 
\end{align}
\begin{align}
\Leftarrow 
\frac{\tilde{W}_a[e-1] +  \frac{\tilde{W}_a[e-1]}{y}  f \eta}{\tilde{W}_a[e-1]} < \frac{\tilde{W}_u[e-1] + \left\lfloor \frac{\tilde{W}_a[e-1]}{y} \right\rfloor (1-f) \eta + \left(  \frac{\tilde{W}_u[e-1] - \tau}{y}  -  \frac{\tilde{W}_a[e-1]}{y}  \right) \eta}{\tilde{W}_u[e-1]} \notag \\
\Leftarrow 1 + \frac{f\eta }{y} < 1 + \left\lfloor \frac{\tilde{W}_a[e-1]}{y} \right\rfloor \frac{(1-f) \eta}{\tilde{W}_u[e-1]} + \frac{1-\frac{\tau}{\tilde{W}_u[e-1]}}{y}\eta - \frac{\tilde{W}_a[e-1]}{\tilde{W}_u[e-1] y} \eta \notag \\
\Leftarrow \frac{f\eta }{y} + \frac{\sigma \eta }{y}  < \left\lfloor \frac{\tilde{W}_a[e-1]}{y} \right\rfloor \frac{(1-f) \eta}{\tilde{W}_u[e-1]} + \frac{1-\frac{\tau}{\tilde{W}_u[e-1]}}{y}\eta \notag \\
\Leftarrow \frac{f\eta }{y} + \frac{\sigma \eta }{y} < \left( \frac{\tilde{W}_a[e-1]}{y} - 1 \right) \frac{(1-f) \eta}{\tilde{W}_u[e-1]} + \frac{1-\frac{\tau}{\tilde{W}_u[e-1]}}{y}\eta \notag 
\end{align}
\begin{align}
\Leftarrow f + \sigma < \frac{(\tilde{W}_a[e-1] - y)(1-f)}{\tilde{W}_u[e-1]} + 1-\frac{\tau}{\tilde{W}_u[e-1]} \notag \\
\Leftarrow f\left(1 + \frac{\tilde{W}_a[e-1]-y}{\tilde{W}_u[e-1]} \right) < \frac{\tilde{W}_a[e-1]-y}{\tilde{W}_u[e-1]} - \sigma + 1-\frac{\tau}{\tilde{W}_u[e-1]} \notag \\
\Leftarrow f \left(1 + \sigma - \frac{y}{\tilde{W}_u[e-1]} \right) < \frac{\tilde{W}_a[e-1]-y}{\tilde{W}_u[e-1]} - \sigma + 1-\frac{\tau}{\tilde{W}_u[e-1]} \notag \\
\Leftarrow f \left(1 + \sigma - \frac{y}{\tilde{W}_u[e-1]} \right)  < - \sigma + 1-\frac{\tau}{\tilde{W}_u[e-1]} \notag \\
\Leftarrow f < \frac{1 - \sigma - \frac{\tau}{\tilde{W}_u[e-1]}}{1 + \sigma - \frac{y}{\tilde{W}_u[e-1]}} \notag \\
\Leftarrow f < \frac{1 - \sigma - \epsilon}{1 + \sigma}
\end{align}
where we have used $\tilde{W}_a[e-1] > y^2/(\eta \epsilon) > y$ and $\tau < \epsilon \tilde{W}_u[e-1] $.

The proof  follows by induction. 
\end{proof}

\begin{lem}
The total wealth of an adversary, following $\pi_a$, after $k$ epochs is upper bounded as 
\begin{align}
\tilde{W}_a[k] \leq W_a[0]\left(1 + \frac{f\eta}{y} \right)^k. 
\end{align}
\label{lem:advwealthub}
\end{lem}
\begin{proof}
The total wealth of the adversary at the end of the initial epoch, $\tilde{W}_a[0]$, is $W_a[0]$. 
Under policy $\pi_a$, the adversary is able to front run $\lfloor \tilde{W}_a[0]/y \rfloor$ many user transactions in the first epoch. 
Therefore, the total wealth of the adversary at the end of the first epoch is (under what assumptions?)  
\begin{align}
\tilde{W}_a[1] = W_a[0] + \left\lfloor \frac{W_a[0]}{y} \right\rfloor f \eta \leq
W_a[0]\left( 1 +  \frac{f \eta}{y} \right).  \label{eq:waupboundsdfx}
\end{align}
From Equation~\eqref{eq:waupboundsdfx}, the total wealth of the adversary at the end of the second epoch is at most 
\begin{align}
\tilde{W}_a[2] \leq W_a[0] \left( 1 + \frac{f \eta}{y} \right) + \left\lfloor \frac{W_a[0] \left( 1 + \frac{f \eta}{y} \right)}{y} \right\rfloor f \eta \leq W_a[0] \left(1 + \frac{f\eta}{y} \right)^2.  \label{eq:totwealthadv2ndepoc}
\end{align}
Continuing this way, the total wealth of the adversary at the of the $k$-th epoch is at most 
\begin{align}
\tilde{W}_a[k] \leq W_a[0] \left(1 + \frac{f \eta}{y} \right)^k. 
\end{align}
\end{proof}

\begin{lemma}
The total wealth of the user at the end of $k$ epochs is lower bounded as 
\begin{align}
\tilde{W}_u[k] \geq W_u[0]\left( 1+\frac{\eta}{y} - \frac{\sigma f \eta}{y}\right)^k -  \frac{\tau \eta}{y} \left( \frac{\left( 1 + \frac{\eta}{y} - \frac{\sigma f \eta}{y}\right)^k - 1}{\left(  \frac{\eta}{y} - \frac{\sigma f \eta }{y}\right)} \right).
\end{align}
\label{lem:userwealtlb}
\end{lemma}
\begin{proof}
Next, the total wealth of the user at the end of the initial epoch is $W_u[0]$. 
During the first epoch, the user makes $\left\lceil \frac{W_u[0]-\tau}{y} \right\rceil$ transactions bringing $\left\lceil \frac{W_u[0]-\tau}{y} \right\rceil \eta$ additional value in to the system. 
From Equation~\eqref{eq:waupboundsdfx}, the total wealth of the user at the end of the first epoch is 
\begin{align}
\tilde{W}_u[1] &\geq W_u[0] + W_a[0] + \left\lceil \frac{W_u[0]-\tau}{y} \right\rceil \eta - W_a[0]\left(1 + \frac{f\eta}{y} \right) \notag \\
&\geq W_u[0] + W_a[0] + \left( \frac{W_u[0]-\tau}{y} \right) \eta - W_a[0]\left(1 + \frac{f\eta}{y} \right) \notag \\
&= W_u[0]\left(1 + \frac{\eta}{y} \right) - W_a[0] \frac{f\eta}{y} - \frac{\tau \eta}{y} \notag \\
&\geq W_u[0]\left(1 + \frac{\eta}{y}\right) - W_u[0]\sigma \frac{f\eta}{y} - \frac{\tau \eta}{y} \notag \\
&= W_u[0]\left(1+\frac{\eta}{y} - \frac{\sigma f \eta}{y}\right) - \frac{\tau\eta}{y}.\label{eq:tildew1lb42}
\end{align}
Similarly, the total wealth of the user at the end of the second epoch is 
\begin{align}
\tilde{W}_u[2] &\geq \tilde{W}_u[1] + \tilde{W}_a[1] + \left\lceil \frac{\tilde{W}_u[1]-\tau}{y} \right\rceil \eta  - \tilde{W}_a[1]  \left( 1 + \frac{f\eta}{y} \right)\notag \\
&\geq \tilde{W}_u[1]\left( 1+\frac{\eta}{y} - \frac{\sigma f \eta}{y}\right) - \frac{\tau\eta}{y} \notag \\
&\geq W_u[0] \left( 1+\frac{\eta}{y} - \frac{\sigma f \eta}{y}\right)^2 - \frac{\tau \eta}{y}\left( 1+\frac{\eta}{y} - \frac{\sigma f \eta}{y} + 1 \right).
\end{align}
Repeating the above for $k$ times, we have 
\begin{align}
\tilde{W}_u[k] \geq W_u[0]\left( 1+\frac{\eta}{y} - \frac{\sigma f \eta}{y}\right)^k -  \frac{\tau \eta}{y} \left( \frac{\left( 1 + \frac{\eta}{y} - \frac{\sigma f \eta}{y}\right)^k - 1}{\left(  \frac{\eta}{y} - \frac{\sigma f \eta }{y}\right)} \right).
\end{align}
\end{proof}

\begin{theorem}
Over epochs, the percentage of user transactions front run in each epoch goes to zero.
\end{theorem}
\begin{proof}
The number of user transactions front run during epoch $e$ is $\lfloor \tilde{W}_a[e-1]/y \rfloor$. 
The total number of transactions made by the user during epoch $e$ is $\lceil (\tilde{W}_u[e]-\tau) / y\rceil$. 
Since each of the adversary's tokens can be used to front run a transaction, the fraction of transactions that are front run during the epoch is given by 
\begin{align}
\frac{\lfloor \tilde{W}_a[e-1]/y \rfloor}{\lceil (\tilde{W}_u[e-1]-\tau) / y\rceil} \leq \frac{\tilde{W}_a[e-1]/y}{(\tilde{W}_u[e-1]-\tau) / y} = \frac{\tilde{W}_a[e-1]}{(\tilde{W}_u[e-1]-\tau)}. 
\end{align}
From Lemma~\ref{lem:advwealthub} and~\ref{lem:userwealtlb} we have the percentage of front run transactions as at most 
\begin{align}
\frac{W_a[0] \left(1 + \frac{f \eta}{y} \right)^{e-1}}{W_u[0]\left( 1+\frac{\eta}{y} - \frac{\sigma f \eta}{y}\right)^{e-1} -  \frac{\tau \eta}{y} \left( \frac{\left( 1 + \frac{\eta}{y} - \frac{\sigma f \eta}{y}\right)^{e-1} - 1}{\left(  \frac{\eta}{y} - \frac{\sigma f \eta }{y}\right)} \right) - \tau}. \label{eq:percentagefrontrunbound}
\end{align}
Taking limit of Equation~\eqref{eq:percentagefrontrunbound} as $e\rightarrow \infty$, we have 
\begin{align}
& \lim_{e\rightarrow \infty} \frac{W_a[0] \left(1 + \frac{f \eta}{y} \right)^{e-1}}{W_u[0]\left( 1+\frac{\eta}{y} - \frac{\sigma f \eta}{y}\right)^{e-1} -  \frac{\tau \eta}{y} \left( \frac{\left( 1 + \frac{\eta}{y} - \frac{\sigma f \eta}{y}\right)^{e-1} - 1}{\left(  \frac{\eta}{y} - \frac{\sigma f \eta }{y}\right)} \right) - \tau} \notag \\
= & \lim_{e\rightarrow \infty} \frac{W_a[0]}{W_u[0]\left( \frac{y+\eta-\sigma f \eta}{y + f\eta} \right)^{e-1} - \frac{\tau}{1-\sigma f} \left( \left( \frac{y+\eta-\sigma f \eta}{y + f\eta} \right)^{e-1} - \frac{y^{e-1}}{(y+f\eta)^{e-1}} \right) - \frac{\tau y^{e-1}}{(y+f\eta)^{e-1}} } \notag \\
= & \lim_{e\rightarrow \infty} \frac{W_a[0]}{\left( W_u[0] - \frac{\tau}{1-\sigma f} \right) \left( \frac{y+\eta-\sigma f \eta}{y + f\eta} \right)^{e-1} + \frac{\tau}{1-\sigma f} \frac{y^{e-1}}{(y+f\eta)^{e-1}}  - \frac{\tau y^{e-1}}{(y+f\eta)^{e-1}} } = 0, 
\end{align}
since $y+\eta - \sigma f \eta > y + f\eta \iff f<1/(1+\sigma)$ which is true, and $W_u[0] > \tau/(1-\sigma f)$. 
\end{proof}

\begin{theorem}
\label{thm:advwealthoptimal}
The total wealth earned by the adversary after $k$ epochs under policy $\pi_a$ is at most a factor $\frac{y}{y-\eta \epsilon} $ away from the total reward under any optimal policy. 
\end{theorem}
(Proof in Appendix~\ref{apx: theorem 4})

\section{Experiments}
We perform extensive experiments in order to show the results of the tokenized transaction protocol under different conditions. The experiments are performed by fixing user policy as described in Algorithm \ref{Algorithm 1} and intelligent adversary policy, as per Algorithm \ref{Algorithm 2}. We run the experiment for 10000 rounds and use the following values of $f=0.8,\eta=100$ for our simulations. We assume that only a single MEV transaction is included in the block. The user starts with an initial wealth of $w_\user=1000\$$ and the adversary starts with $\frac{w_{\user}}{2}$. $y=80$ is our chosen token cost. 

\subsection{Transaction Models}
We compare the results of our experiment, to the current status-quo, that contains transactions that are being ordered by a relay such as flashbots, and provided to a validator. Here, $f=1$, since the transactions are always attacked by the adversary. We also compare our results to an ideal scenario, where no MEV transaction is attacked. We use the final wealth of the user as an indicator to determine the long term feasibility of our policy. We also consider the percentage of transactions that are frontrun by the adversary, in order to track the mitigation of the MEVs from our method. We consider the wealth of all parties in the system based on the following model configurations:

\begin{itemize}
    \item Constant $\eta$: This state is when each block only has a single MEV transaction, and the user and adversary gain a constant, deterministic reward $\eta$ for making this transaction. If an adversary attacks the transaction, they earn a constant reward of $f\eta$, while the user earns $\eta-f\eta$.    
    \item Stochastic $\eta$: This state is when each block has a single MEV transaction, but the user and adversary gain a reward that is stochastic. 
    \item Fatal frontrunning: While making MEV transactions, there are some attacks that cannot be blocked if they are frontrun, this is called fatal frontrunning. In this case, we include Type 1 and Type 2 transactions, where Type 1 are those transactions that can be protected (such as swaps), and Type 2 are those that cannot (such as arbitrage attacks and liquidation).
    \item Real world $\eta$: In this case, we consider real values of MEV profits that have been extracted from the Ethereum blockchain. 
\end{itemize}

\subsection{Tokenization on constant $\eta$}
In this case, we consider the wealth gained by the user and the adversary based on constant $\eta=100$, as shown in Figure \ref{fig:const-eta}. 
We see in Figure \ref{fig:const-eta}, that the total user wealth, while using Masquerade is comparable to the ideal scenario, one in which the adversary does not exist. The status-quo represents the current protocol, where all MEV transactions are frontrun by the adversary. We see, that on average, only 30\% of transactions are attacked, which is a significant improvement, compared to all transactions being attacked by adversary. We see, that when the user waits for a certain period of time, they are able to beat the adversary and protect their transaction.

\begin{figure}
    \centering
    \includegraphics[width=0.45\linewidth]{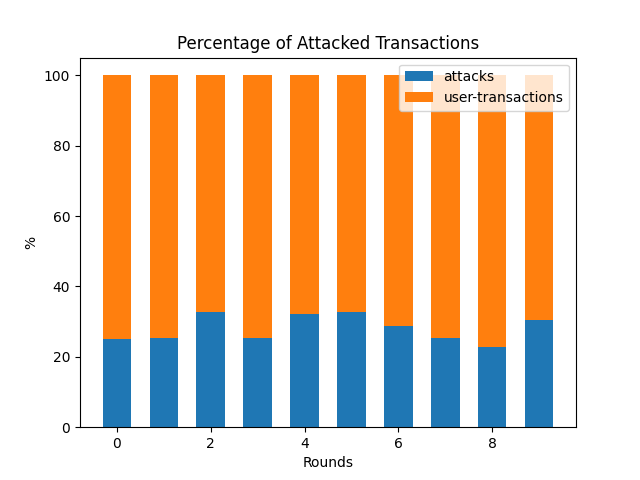}
    \includegraphics[width=0.45\linewidth]{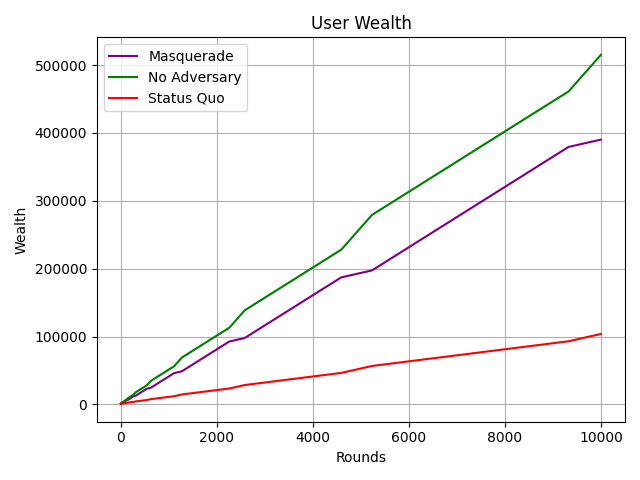}
     \caption{(a) Shows the percentage of transactions successfully attacked by the adversary every 1000 rounds. (b) The accumulation of the wealth of users and the adversaries when using Masquerade (our method) v/s the current status quo.}
    \label{fig:const-eta}
\end{figure}

\begin{table}[]
    \begin{tabular}{|c|c|c|c|c|c|c|}
    \toprule
    Method & \multicolumn{2}{c}{Wealth of User} & \multicolumn{2}{c}{ Wealth of Adversary}&\% Frontrun&\% Backrun \\   
    
     &Initial&Final&Initial&Final&&\\
     \hline
      Status Quo & 1000&103840.0& 500 & 411560.0&100&100 \\
      Token MEV & 1000&390240.0& 500 & 125160.0&29.81&30.90\\
      Ideal case &1000&515400.0&0&0&0&0\\
    \bottomrule
    \end{tabular}
    \caption{Improvements in transaction unstealibility using tokens}
    \label{tab:wealth}
\end{table}

\subsection{Tokenization on stochastic $\eta$}
In this case, we consider the wealth gained by the user and the adversary for a stochastic $\eta$ based on samples drawn from a Gaussian distribution, as shown in Figure \ref{fig:stochastic-eta}. 
We also consider a heavy tailed Cauchy distribution as shown in Figure \ref{fig:stochastic-eta}. This is a more interesting case, as the value earned by stealing the transaction is vastly different. The adversary now, no longer can predict what value each MEV transaction that can be potentially made in the future holds. They are only aware of current token request transactions, and tokenized MEV transactions. Now, the adversary has to carefully decide if they would like to use their best token to attack the current tokenized MEV transaction, or wait in case a better prospect in the future shows up. If the adversary chooses not to attack, the user benefits directly. If the adversary attacks, they may lose out on a future transaction. To capture this, we consider simple modifications to Algorithm 1 where the user decides to use their lowest token if the value of $\eta>100$, and second lowest token otherwise. Similarly, in Algorithm 2, the adversary decides whether to use their best token in the current round, or the next.  

\begin{figure}[]
    \centering
    \includegraphics[width=0.45\linewidth]{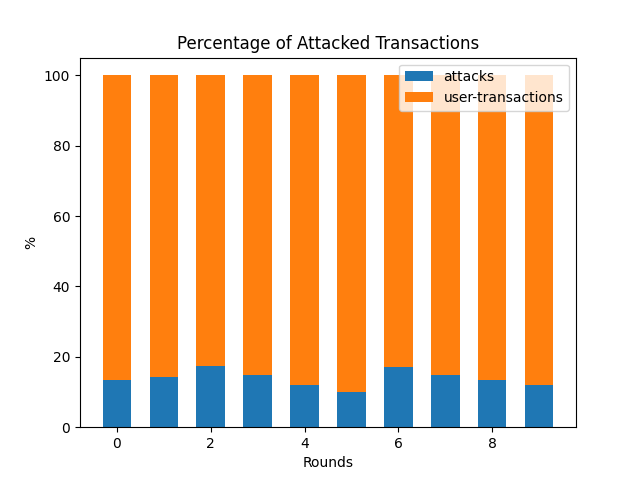}
    \includegraphics[width=0.45\linewidth]{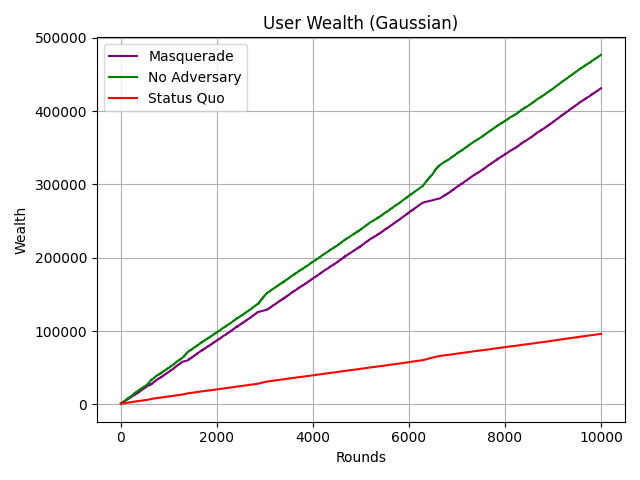}
    \includegraphics[width=0.45\linewidth]{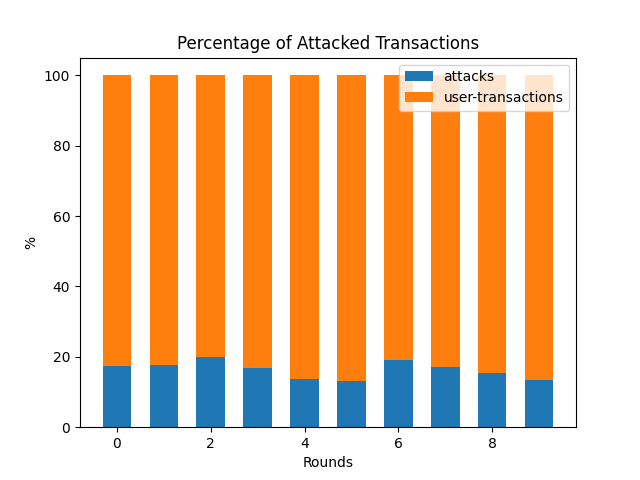}
    \includegraphics[width=0.45\linewidth]{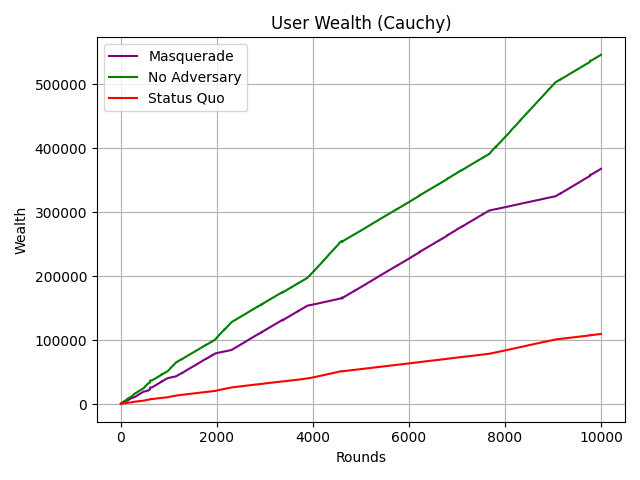}    
    \caption{The user and adversary wealths for stochastic values of eta, drawn from Gaussian and Cauchy distributions}
    \label{fig:stochastic-eta}
\end{figure}

Here too, we see that our algorithm performs quite well, improving the user's wealth compared to current status quo. We also see that the adversary is unable to attack all transactions, and even misses out on crucial, high value MEVs. The challenge for the adversary here, lies in guessing the most profitable transactions, and being able to attack them successfully. If the adversary waits for too long in a phase without using tokens, then the adversary misses out on purchasing early tokens in the next phase. The user in the next phase can generally exploit this and use later tokens for low value transactions and early tokens for high value transactions. We can consider several other efficient user policies in this case. 

\subsection{Fatal frontrunning attacks}
In a fatal frontrunning attack, the attacker's transaction executes while the user's transaction fails. This happens mostly in case of arbitrages, or loan liquidations. If the user discovers an arbitrage opportunity and issues a transaction, the adversary can observe the user's transaction, duplicate it, and issue its own transaction with a higher gas fee. In this case, the adversary's transaction executes, while the user receives no profit. No matter what slippage the user specifies, the adversary has gained the full profit. This leads us to two types of transactions, the first kind, or type 1 transactions, which involve exchanges that can be controlled with $f,\eta$, and type 2 transactions, that are fatal, such as liquidation attacks and arbitrage attacks. Creating an MEV transaction is a computationally expensive operation for the user. The more compute the user can afford, the higher value MEV transactions they can find. Thus, here, if the user sees that they are not getting adequate rewards for the compute spent, the user would try to find computationally cheaper MEV transactions, which for our experiments is equivalent to low $\eta$ transactions. Thus, for type 2 transactions, the user only engages with small $\eta$.

We ran our experiment on attacks that include fatal frontrunning attacks, and we see that in such cases, we are able to only prevent a very small percentage of these attacks, however even so, the total wealth obtained by the user remains better than compared to the status quo. For fatal frontrunning attacks that happened about 50\% of the time, users get successfully attacked 70\% of the time. This number decreased when we slightly modified Algorithm \ref{Algorithm 1}. In this case, the user only engaged with $\eta<100$ for type 2 transactions, and thus, the percentage of frontruns decreased significantly. In Figure \ref{figF50}, we include fatal frontrunning attacks both 50\% and 30\% of the time, and observe that users are able to increase their wealth successfully.  

\begin{figure}
    \centering
    \includegraphics[width=0.45\linewidth]{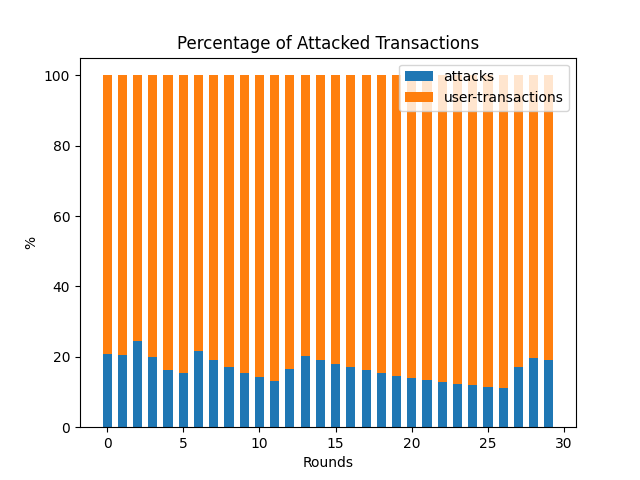}
    \includegraphics[width=0.45\linewidth]{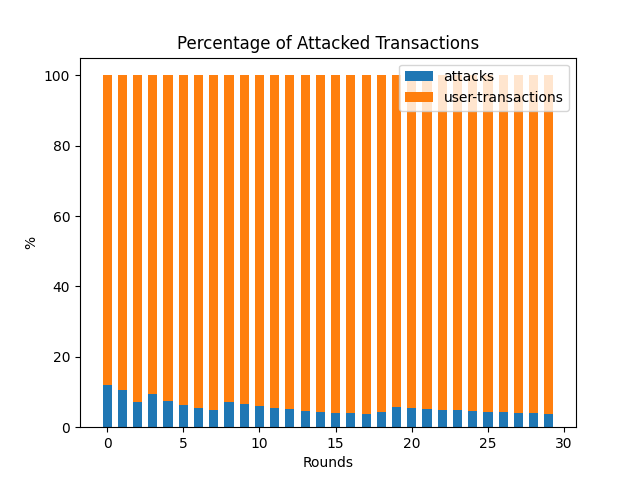}
    \includegraphics[width=0.45\linewidth]{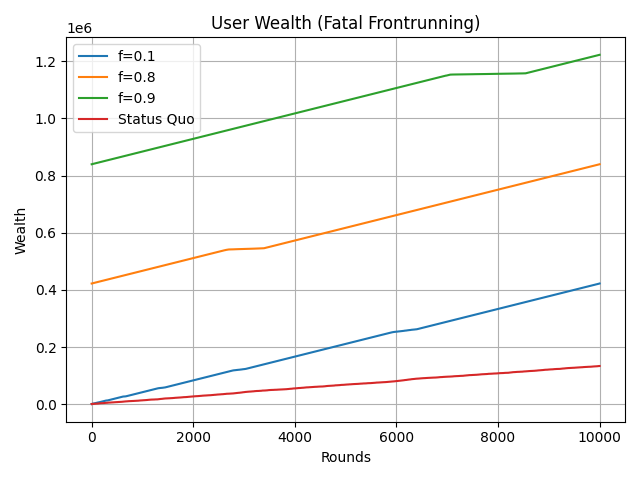}
    \includegraphics[width=0.45\linewidth]{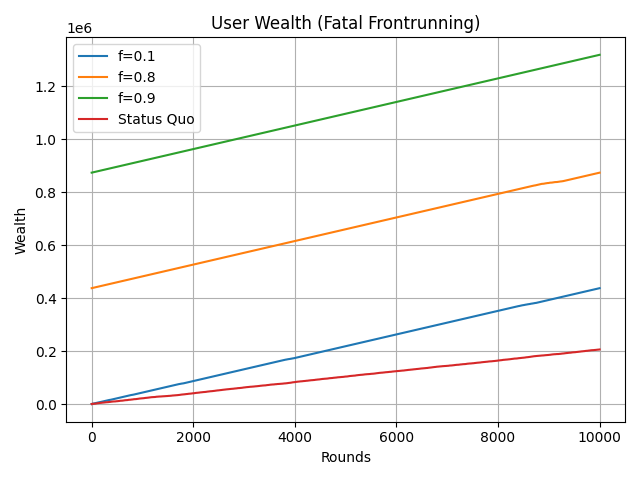}

       \caption{The wealth of the adversary, when including fatal frontrunning attacks (a)50\%  and (b)30\% of the time, for f=0.1,0.8 and 0.9.}
    \label{figF50}
\end{figure}

\subsection{Tokenization on real-world MEV rewards}
Here, we consider the real world MEV values that we have extracted from the Ethereum blockchain. We use a similar procedure as when $\eta$ is stochastic. Here too, we see that masquerade performs better than current practices. Figure \ref{fig:real-mev} shows the total wealth of the users in this system, along with successful attacks. 
\begin{figure}
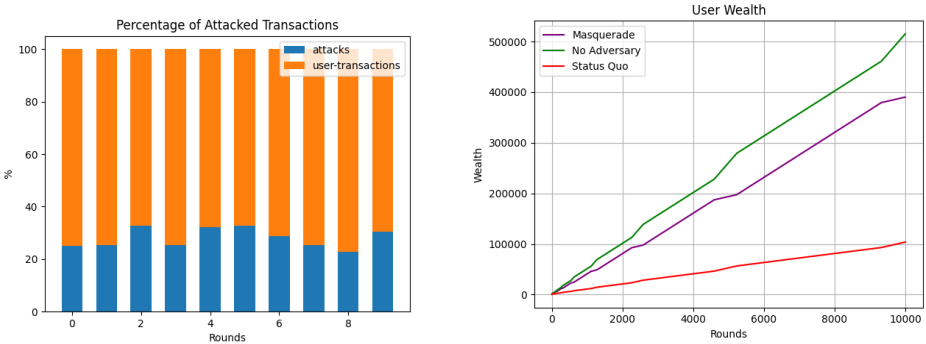

    \centering
    \includegraphics[width=0.45\linewidth]{per-frontruns-real-mev.png}
    \includegraphics[width=0.45\linewidth]{User-wealth-real-mev.png}
    \caption{Masquerade performs quite well even on real world Ethereum data}
    \label{fig:real-mev}
\end{figure}

\subsection{Token Epochs}
For ease in analysis, we consider an alternate arrangement of the current proposed model, as described in Section 5, where we separate the token purchase and token spending into epochs. We can see that this process of separating the token purchase and spend into phases is the same as our method, where the tokens are purchased and spent continuously, over the course of a number of rounds, as visualized in Figure \ref{PhaseComp}. 

\begin{figure}
    \centering
    \includegraphics[width=0.5\linewidth]{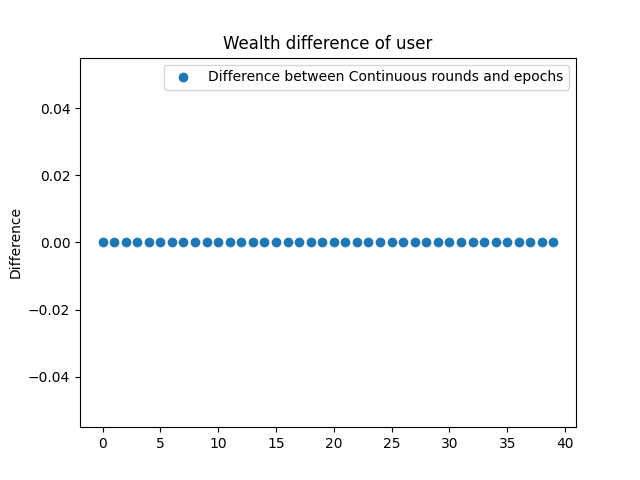}
    \caption{Here, we see that the user wealth and the adversary wealth are equal at different rounds/epochs, they are simply shifted, and their difference is constant.}
    \label{PhaseComp}
\end{figure}

We observe that the continuous token purchase lags by a number of rounds. This is because in the continuous case we also consider the possibility of a user wanting to make non-MEV transactions, and the adversary does not attack these transactions. As a result, there is a slight delay in reaching the same amounts of wealth. However, they are essentially the same.

\subsection{Ablation studies}

We perform multiple ablation studies using different token conditions, such as different token costs, expiration of the tokens, scenarios where the token cost is unable to be returned to the users and the adversaries in the following scenarios:

\subsubsection{Token Costs}
We consider different scenarios for token costs, ranging from extremely cheap tokens, to extremely expensive tokens. As the tokens are eventually returned to the users in the system, the token costs do not influence total wealth of the system. If the tokens are cheaper, more users will be able to afford them, whereas if the tokens are expensive, fewer users can afford them, which reduces the number of tokenized transactions that can be made in a single round. As a result, the adversary is able to manipulate this result due to all users being able to afford only small number of tokens. We see that even when starting with a small sum of wealth, the adversary is unable to frontrun or backrun the user, when the token costs are reasonably priced.


\subsubsection{Token with Expiration}
In this case, the tokens expire after a set number of rounds. If the user does not use these tokens within the set limit, they expire and the users are unable to recoup the costs spent on the token. As a consequence of this, we observe that it is generally in best interest of the user and adversary to not hoard tokens and spend tokens before they expire.

\subsubsection{Token with no Refund}
In this case, a user will not receive a refund when they make a tokenized transaction. Here, the token costs need to be cheap in order to incentivize users to make tokenized transactions.

\section{Related Work}
Common types of MEVs include frontrunning, backrunning and sandwich attacks. An adversary is generally able to frontrun the transaction by listing a higher transaction fee, so that their transaction is executed before the honest user. Backrunning involves a similar strategy, except the adversary's transaction is placed immediately after the user. Sometimes, an adversary can place a pair of transactions, right before and right after the user transaction, which can lead to an artificial manipulation of asset prices, resulting in profit at the cost of the user.

These MEVs, pose a significant threat to users, cause network congestion and may even lead to centralization without appropriate mitigation strategies. So far, there have been efforts to introduce countermeasures to MEVs. In this section, we list the most popular solutions to MEV as follows:
\begin{itemize}
    \item MEV Auction Platforms: These MEV auction platforms are commonly used in post-merge Ethereum. Generally, the idea here is that there exist builders such as Flashbots\cite{noauthor_flashbots_nodate}, BloXroute\cite{klarman2018bloxroute}, MEV-Boost\cite{yang2022sok} who assemble blocks with transactions they receive from users, the public mempool as well as the ones the builder itself inserts to generate MEV. A user is allowed to specify their transaction ordering preferences to miners in exchange for compensation. These assembled blocks are then submitted to a block proposer, with some promised profits. This is a "commit-and-reveal" strategy, i.e the proposer is unable to see the order of the transactions or block contents, they can only look at the promised profit and some other metadata. The block proposer then selects the block they wish to propose. 
    MEV auction platforms attempt to create a more competitive and efficient environment for users to have control over how their transactions are ordered and to capture some of the value that would otherwise go to miners. These platforms can be beneficial for users who want to ensure their transactions are executed in a specific way or who want to benefit from MEV themselves. A key issue with this method is the reliance on a "trusted builder", though theoretically, any builder can be part of the network. 

    \item Time-based Ordering Solutions: These solutions rely on establishing a certain order in transactions, they establish a set of properties that the transaction needs to satisfy in order to be part of the block. An example of this class of solutions is Hedera\cite{baird2019hedera}, that uses median to establish the timestamp of the received transaction. Fair ordering means that $\tx_u$ is executed before $\tx_a$ if the median received time of $\tx_u$<median received time of $\tx_a$. A problem of using median timestamps in this way, is that it is susceptible to manipulation by adversary. Other examples of receive-order fairness are Themis\cite{kelkar2021themis} and Aequitas\cite{kelkar2020order} which consider that if atleast $\gamma$ of the nodes receive transaction $\tx_u$ before $\tx_a$, then $\tx_u$ can be included no later than $\tx_a$. Wendy\cite{kursawe2020wendy} uses relative fairness to approach this issue. Their solution, is that if there exists some time $t$ at which all honest nodes have seen transaction $\tx_u$, and they then saw transaction $\tx_a$ after $t$, then $\tx_u$ is executed before $\tx_a$. These class of solutions pose a risk to the users based on latencies of nodes. 

    \item Content-agnostic Ordering: These class of solutions do not impose a constraint on the ordering of the transactions, as long as the ordering is independent of the content of these transactions. Most of these algorithms generally encrypt all transactions, or use a trusted third party in order to hide the transactions. They then wait for these transactions to be committed on to the Blockchain, and finally provide the secret key in order to reveal the transaction, and ensure validity etc. An example of this is TEX\cite{khalil2019tex}, which is used in a proof of work situation, where the user encrypts their transaction using timelock puzzles, with the understanding that the attacker cannot solve the puzzles faster than the user. This uses a trusted custodian, and thus requires users to regularly check in to make sure the custodian is not misbehaving. Tesseract\cite{bentov2019tesseract} requires the user to trust additional secure hardware called TEE (Trusted Execution Environment), which is responsible for encrypting the transactions and releasing them after they have been committed to the block. Most of these solutions remove ordering privileges from the miners, but result in additional trust assumptions that may reduce the decentralization in the network. 
\end{itemize}

\section{Conclusion}
In this paper, we have shown that the utility obtained by the tokenization system is better than the current utility, and we have suggested a possible reform to the system without the need of a trusted builder or proposer. 
The token system is a very powerful way for the user to be able to trick the adversary in the real world, as an adversary is faced with the choice of spending their token or saving it for a later transaction that may of may not provide a better value. An adversary only has the information that a token has been purchased, and is unable to decipher which token can be used for which transaction, until it is too late. We have shown a huge improvement in the transaction unstealability for a user in the system, thus improving the overall fairness in the network. 
One of the major drawbacks of Masquerade, is that it does not completely eliminate the presence of MEV attacks. To a large extent, our algorithm prevents user losses, and stops transaction order manipulation by removing a significant amount of miner privileges, however, blind frontrunning is still possible in the system. An extension to the project involves looking for lightweight, scalable solutions that can build on top of Masquerade in order to completely eliminate MEV attacks. One possible extension to this project also is to include dynamic behaviors by the adversary and user in response to current policies.





\bibliographystyle{ACM-Reference-Format}
\bibliography{sample-base}


\begin{thebibliography}{25}


\ifx \showCODEN    \undefined \def \showCODEN     #1{\unskip}     \fi
\ifx \showDOI      \undefined \def \showDOI       #1{#1}\fi
\ifx \showISBNx    \undefined \def \showISBNx     #1{\unskip}     \fi
\ifx \showISBNxiii \undefined \def \showISBNxiii  #1{\unskip}     \fi
\ifx \showISSN     \undefined \def \showISSN      #1{\unskip}     \fi
\ifx \showLCCN     \undefined \def \showLCCN      #1{\unskip}     \fi
\ifx \shownote     \undefined \def \shownote      #1{#1}          \fi
\ifx \showarticletitle \undefined \def \showarticletitle #1{#1}   \fi
\ifx \showURL      \undefined \def \showURL       {\relax}        \fi
\providecommand\bibfield[2]{#2}
\providecommand\bibinfo[2]{#2}
\providecommand\natexlab[1]{#1}
\providecommand\showeprint[2][]{arXiv:#2}

\bibitem[noa(2023a)]%
        {noauthor_binance_nodate}
 \bibinfo{year}{2023}\natexlab{a}.
\newblock \bibinfo{title}{Binance - {Cryptocurrency} {Exchange} for {Bitcoin},
  {Ethereum} \& {Altcoins}}.
\newblock
\newblock
\urldef\tempurl%
\url{https://www.binance.com}
\showURL{%
\tempurl}


\bibitem[def(2023)]%
        {defisurge}
 \bibinfo{year}{2023}\natexlab{}.
\newblock \bibinfo{title}{Decentralized Finance Market to Surge}.
\newblock
\newblock
\urldef\tempurl%
\url{https://finance.yahoo.com/news/decentralized-finance-market-surge-usd-131500988.html}
\showURL{%
\tempurl}


\bibitem[noa(2023b)]%
        {noauthor_flashbots_nodate}
 \bibinfo{year}{2023}\natexlab{b}.
\newblock \bibinfo{title}{Flashbots}.
\newblock
\newblock
\urldef\tempurl%
\url{https://www.flashbots.net}
\showURL{%
\tempurl}


\bibitem[mev(2023)]%
        {mevoutlook}
 \bibinfo{year}{2023}\natexlab{}.
\newblock \bibinfo{title}{MEV Outlook 2023}.
\newblock
\newblock
\urldef\tempurl%
\url{https://eigenphi.substack.com/p/mev-outlook-2023#:~:text=In%202022%2C%20the%20total%20MEV,%2472.8%20billion%20contributed%20by%20bots}
\showURL{%
\tempurl}


\bibitem[noa(2023c)]%
        {noauthor_pancakeswap_nodate}
 \bibinfo{year}{2023}\natexlab{c}.
\newblock \bibinfo{title}{Pancake Swap - {Cryptocurrency} {Exchange} for
  {Ethereum}}.
\newblock
\newblock
\urldef\tempurl%
\url{https://pancakeswap.finance/}
\showURL{%
\tempurl}


\bibitem[Adams et~al\mbox{.}(2021)]%
        {adams2021uniswap}
\bibfield{author}{\bibinfo{person}{Hayden Adams}, \bibinfo{person}{Noah
  Zinsmeister}, \bibinfo{person}{Moody Salem}, \bibinfo{person}{River Keefer},
  {and} \bibinfo{person}{Dan Robinson}.} \bibinfo{year}{2021}\natexlab{}.
\newblock \showarticletitle{Uniswap v3 core}.
\newblock \bibinfo{journal}{\emph{Tech. rep., Uniswap, Tech. Rep.}}
  (\bibinfo{year}{2021}).
\newblock


\bibitem[Baird et~al\mbox{.}(2019)]%
        {baird2019hedera}
\bibfield{author}{\bibinfo{person}{Leemon Baird}, \bibinfo{person}{Mance
  Harmon}, {and} \bibinfo{person}{Paul Madsen}.}
  \bibinfo{year}{2019}\natexlab{}.
\newblock \showarticletitle{Hedera: A public hashgraph network \& governing
  council}.
\newblock \bibinfo{journal}{\emph{White Paper}}  \bibinfo{volume}{1}
  (\bibinfo{year}{2019}), \bibinfo{pages}{9--10}.
\newblock


\bibitem[Bentov et~al\mbox{.}(2019)]%
        {bentov2019tesseract}
\bibfield{author}{\bibinfo{person}{Iddo Bentov}, \bibinfo{person}{Yan Ji},
  \bibinfo{person}{Fan Zhang}, \bibinfo{person}{Lorenz Breidenbach},
  \bibinfo{person}{Philip Daian}, {and} \bibinfo{person}{Ari Juels}.}
  \bibinfo{year}{2019}\natexlab{}.
\newblock \showarticletitle{Tesseract: Real-time cryptocurrency exchange using
  trusted hardware}. In \bibinfo{booktitle}{\emph{Proceedings of the 2019 ACM
  SIGSAC Conference on Computer and Communications Security}}.
  \bibinfo{pages}{1521--1538}.
\newblock


\bibitem[Brennecke et~al\mbox{.}(2022)]%
        {brennecke2022central}
\bibfield{author}{\bibinfo{person}{Martin Brennecke}, \bibinfo{person}{Tobias
  Guggenberger}, \bibinfo{person}{Benjamin Schellinger}, {and}
  \bibinfo{person}{Nils Urbach}.} \bibinfo{year}{2022}\natexlab{}.
\newblock \showarticletitle{The de-central bank in decentralized finance: a
  case study of MakerDAO}.
\newblock  (\bibinfo{year}{2022}).
\newblock


\bibitem[Cachin et~al\mbox{.}(2022)]%
        {cachin2022quick}
\bibfield{author}{\bibinfo{person}{Christian Cachin}, \bibinfo{person}{Jovana
  Mi{\'c}i{\'c}}, \bibinfo{person}{Nathalie Steinhauer}, {and}
  \bibinfo{person}{Luca Zanolini}.} \bibinfo{year}{2022}\natexlab{}.
\newblock \showarticletitle{Quick order fairness}. In
  \bibinfo{booktitle}{\emph{International Conference on Financial Cryptography
  and Data Security}}. Springer, \bibinfo{pages}{316--333}.
\newblock


\bibitem[Chitra and Kulkarni(2022)]%
        {chitra2022improving}
\bibfield{author}{\bibinfo{person}{Tarun Chitra} {and} \bibinfo{person}{Kshitij
  Kulkarni}.} \bibinfo{year}{2022}\natexlab{}.
\newblock \showarticletitle{Improving proof of stake economic security via mev
  redistribution}. In \bibinfo{booktitle}{\emph{Proceedings of the 2022 ACM CCS
  Workshop on Decentralized Finance and Security}}. \bibinfo{pages}{1--7}.
\newblock


\bibitem[Heimbach et~al\mbox{.}(2023)]%
        {heimbach2023ethereum}
\bibfield{author}{\bibinfo{person}{Lioba Heimbach}, \bibinfo{person}{Lucianna
  Kiffer}, \bibinfo{person}{Christof~Ferreira Torres}, {and}
  \bibinfo{person}{Roger Wattenhofer}.} \bibinfo{year}{2023}\natexlab{}.
\newblock \showarticletitle{Ethereum's Proposer-Builder Separation: Promises
  and Realities}.
\newblock \bibinfo{journal}{\emph{arXiv preprint arXiv:2305.19037}}
  (\bibinfo{year}{2023}).
\newblock


\bibitem[Heimbach and Wattenhofer(2022a)]%
        {heimbach2022eliminating}
\bibfield{author}{\bibinfo{person}{Lioba Heimbach} {and} \bibinfo{person}{Roger
  Wattenhofer}.} \bibinfo{year}{2022}\natexlab{a}.
\newblock \showarticletitle{Eliminating sandwich attacks with the help of game
  theory}. In \bibinfo{booktitle}{\emph{Proceedings of the 2022 ACM on Asia
  Conference on Computer and Communications Security}}.
  \bibinfo{pages}{153--167}.
\newblock


\bibitem[Heimbach and Wattenhofer(2022b)]%
        {heimbach2022sok}
\bibfield{author}{\bibinfo{person}{Lioba Heimbach} {and} \bibinfo{person}{Roger
  Wattenhofer}.} \bibinfo{year}{2022}\natexlab{b}.
\newblock \showarticletitle{Sok: Preventing transaction reordering
  manipulations in decentralized finance}.
\newblock \bibinfo{journal}{\emph{arXiv preprint arXiv:2203.11520}}
  (\bibinfo{year}{2022}).
\newblock


\bibitem[Kazerani et~al\mbox{.}(2017)]%
        {kazerani2017determining}
\bibfield{author}{\bibinfo{person}{Ali Kazerani}, \bibinfo{person}{Domenic
  Rosati}, {and} \bibinfo{person}{Brian Lesser}.}
  \bibinfo{year}{2017}\natexlab{}.
\newblock \showarticletitle{Determining the usability of bitcoin for beginners
  using change tip and coinbase}. In \bibinfo{booktitle}{\emph{Proceedings of
  the 35th ACM International Conference on the Design of Communication}}.
  \bibinfo{pages}{1--5}.
\newblock


\bibitem[Kelkar et~al\mbox{.}(2021)]%
        {kelkar2021themis}
\bibfield{author}{\bibinfo{person}{Mahimna Kelkar}, \bibinfo{person}{Soubhik
  Deb}, \bibinfo{person}{Sishan Long}, \bibinfo{person}{Ari Juels}, {and}
  \bibinfo{person}{Sreeram Kannan}.} \bibinfo{year}{2021}\natexlab{}.
\newblock \showarticletitle{Themis: Fast, strong order-fairness in byzantine
  consensus}.
\newblock \bibinfo{journal}{\emph{Cryptology ePrint Archive}}
  (\bibinfo{year}{2021}).
\newblock


\bibitem[Kelkar et~al\mbox{.}(2020)]%
        {kelkar2020order}
\bibfield{author}{\bibinfo{person}{Mahimna Kelkar}, \bibinfo{person}{Fan
  Zhang}, \bibinfo{person}{Steven Goldfeder}, {and} \bibinfo{person}{Ari
  Juels}.} \bibinfo{year}{2020}\natexlab{}.
\newblock \showarticletitle{Order-fairness for byzantine consensus}. In
  \bibinfo{booktitle}{\emph{Advances in Cryptology--CRYPTO 2020: 40th Annual
  International Cryptology Conference, CRYPTO 2020, Santa Barbara, CA, USA,
  August 17--21, 2020, Proceedings, Part III 40}}. Springer,
  \bibinfo{pages}{451--480}.
\newblock


\bibitem[Khalil et~al\mbox{.}(2019)]%
        {khalil2019tex}
\bibfield{author}{\bibinfo{person}{Rami Khalil}, \bibinfo{person}{Arthur
  Gervais}, {and} \bibinfo{person}{Guillaume Felley}.}
  \bibinfo{year}{2019}\natexlab{}.
\newblock \showarticletitle{Tex-a securely scalable trustless exchange}.
\newblock \bibinfo{journal}{\emph{Cryptology ePrint Archive}}
  (\bibinfo{year}{2019}).
\newblock


\bibitem[Klarman et~al\mbox{.}(2018)]%
        {klarman2018bloxroute}
\bibfield{author}{\bibinfo{person}{Uri Klarman}, \bibinfo{person}{Soumya Basu},
  \bibinfo{person}{Aleksandar Kuzmanovic}, {and} \bibinfo{person}{Emin~G{\"u}n
  Sirer}.} \bibinfo{year}{2018}\natexlab{}.
\newblock \showarticletitle{bloxroute: A scalable trustless blockchain
  distribution network whitepaper}.
\newblock \bibinfo{journal}{\emph{IEEE Internet of Things Journal}}
  (\bibinfo{year}{2018}).
\newblock


\bibitem[Kursawe(2020)]%
        {kursawe2020wendy}
\bibfield{author}{\bibinfo{person}{Klaus Kursawe}.}
  \bibinfo{year}{2020}\natexlab{}.
\newblock \showarticletitle{Wendy, the good little fairness widget: Achieving
  order fairness for blockchains}. In \bibinfo{booktitle}{\emph{Proceedings of
  the 2nd ACM Conference on Advances in Financial Technologies}}.
  \bibinfo{pages}{25--36}.
\newblock


\bibitem[Malkhi and Szalachowski(2022)]%
        {malkhi2022maximal}
\bibfield{author}{\bibinfo{person}{Dahlia Malkhi} {and} \bibinfo{person}{Pawel
  Szalachowski}.} \bibinfo{year}{2022}\natexlab{}.
\newblock \showarticletitle{Maximal extractable value (mev) protection on a
  dag}.
\newblock \bibinfo{journal}{\emph{arXiv preprint arXiv:2208.00940}}
  (\bibinfo{year}{2022}).
\newblock


\bibitem[Pichl and Kaizoji(2017)]%
        {pichl2017volatility}
\bibfield{author}{\bibinfo{person}{Luk{\'a}{\v{s}} Pichl} {and}
  \bibinfo{person}{Taisei Kaizoji}.} \bibinfo{year}{2017}\natexlab{}.
\newblock \showarticletitle{Volatility analysis of bitcoin}.
\newblock \bibinfo{journal}{\emph{Quantitative Finance and Economics}}
  \bibinfo{volume}{1}, \bibinfo{number}{4} (\bibinfo{year}{2017}),
  \bibinfo{pages}{474--485}.
\newblock


\bibitem[Yang et~al\mbox{.}(2022)]%
        {yang2022sok}
\bibfield{author}{\bibinfo{person}{Sen Yang}, \bibinfo{person}{Fan Zhang},
  \bibinfo{person}{Ken Huang}, \bibinfo{person}{Xi Chen},
  \bibinfo{person}{Youwei Yang}, {and} \bibinfo{person}{Feng Zhu}.}
  \bibinfo{year}{2022}\natexlab{}.
\newblock \showarticletitle{Sok: Mev countermeasures: Theory and practice}.
\newblock \bibinfo{journal}{\emph{arXiv preprint arXiv:2212.05111}}
  (\bibinfo{year}{2022}).
\newblock


\bibitem[Zhang et~al\mbox{.}(2022)]%
        {zhang2022flash}
\bibfield{author}{\bibinfo{person}{Haoqian Zhang}, \bibinfo{person}{Louis-Henri
  Merino}, \bibinfo{person}{Vero Estrada-Galinanes}, {and}
  \bibinfo{person}{Bryan Ford}.} \bibinfo{year}{2022}\natexlab{}.
\newblock \showarticletitle{Flash freezing flash boys: Countering blockchain
  front-running}. In \bibinfo{booktitle}{\emph{2022 IEEE 42nd International
  Conference on Distributed Computing Systems Workshops (ICDCSW)}}. IEEE,
  \bibinfo{pages}{90--95}.
\newblock


\bibitem[Zhang et~al\mbox{.}(2020)]%
        {zhang2020byzantine}
\bibfield{author}{\bibinfo{person}{Yunhao Zhang}, \bibinfo{person}{Srinath
  Setty}, \bibinfo{person}{Qi Chen}, \bibinfo{person}{Lidong Zhou}, {and}
  \bibinfo{person}{Lorenzo Alvisi}.} \bibinfo{year}{2020}\natexlab{}.
\newblock \showarticletitle{Byzantine ordered consensus without byzantine
  oligarchy}. In \bibinfo{booktitle}{\emph{14th USENIX Symposium on Operating
  Systems Design and Implementation (OSDI 20)}}. \bibinfo{pages}{633--649}.
\newblock


\end{thebibliography}

\appendix
\section{Proof of Theorem~\ref{thm:advwealthoptimal}}
\label{apx: theorem 4}

\begin{proof}
Supposing $W$ is the total wealth of adversary at the end of an epoch $e$ under any policy $\pi$. 
The total wealth of the adversary at the end of epoch $e+1$ under $\pi$ is at most $W + \frac{W}{y} f\eta$. 
Hence, the total wealth of adversary at end of $k$ epochs is at most $W_a[0]\left( 1 + \frac{f\eta}{y}\right)^k$ under $\pi$. 

Now, since $\pi_a$ is a balanced policy 
\begin{align}
\tilde{W}_a[1] = W_a[0] + \left\lfloor \frac{W_a[0]}{y} \right\rfloor f\eta \geq W_a[0]\left( 1 + \frac{ f \eta}{y} \right) - f\eta.
\end{align}
Similarly, 
\begin{align}
\tilde{W}_a[2] \geq \tilde{W}_a[1]\left( 1 + \frac{f\eta}{y} \right) -f\eta  \geq W_a[0] \left( 1 + \frac{f\eta}{y} \right)^2 -f\eta \left(1 + \frac{f\eta}{y} + 1\right). 
\end{align}
Repeating $k$ times, we have 
\begin{align}
\tilde{W}_a[k] \geq W_a[0] \left( 1 + \frac{f\eta}{y} \right)^k - f\eta \frac{\left(1 + \frac{f\eta}{y} \right)^k - 1}{\frac{f\eta}{y}} \notag \\
=  W_a[0] \left( 1 + \frac{f\eta}{y} \right)^k - y \left(1 + \frac{f\eta}{y} \right)^k + y. 
\end{align}
Therefore the ratio between total wealth achieved under $\pi_a$ and the upper bound is 
\begin{align}
\lim_{k\rightarrow \infty} \frac{W_a[0]\left( 1 + \frac{f\eta}{y}\right)^k}{W_a[0] \left( 1 + \frac{f\eta}{y} \right)^k - y \left(1 + \frac{f\eta}{y} \right)^k + y} = \frac{W_a[0]}{W_a[0] - y} \leq \frac{y}{y-\eta \epsilon}k 
\end{align}
where the last inequality is due to $W_a[0] > y^2/(\eta \epsilon)$. 
\end{proof}

\end{document}